\documentclass[11pt,conference,onecolumn]{IEEEtran}
\usepackage{enumerate}
\usepackage{comment}
\usepackage{amsfonts}
\usepackage{amsmath}
\usepackage{amssymb}
\usepackage{amsthm}
\usepackage{lscape}
\usepackage{epsf}
\usepackage{graphicx}
\usepackage{verbatim}
\usepackage{comment} 
\usepackage{color} 
\newtheorem{theorem}{\indent Theorem}[section]
\newtheorem{lemma}[theorem]{\indent Lemma}

\newtheorem{observation}[theorem]{\indent Observation}
\newtheorem{proposition}[theorem]{\indent Proposition}
\newtheorem{claim}[theorem]{\indent Claim}
\newtheorem{construction}[theorem]{\indent Construction}

\newtheorem{EXAMPLE}{\indent Example}[section]
\newtheorem{definition}{\indent Definition}[section]
\newtheorem{problem}{\indent Problem}[section]

\newenvironment{example}{\begin{EXAMPLE}\rm}{\rm\end{EXAMPLE}}

\newcommand{\code}{{\mathcal{C}}}
\newcommand{\graph}{{\mathcal{G}}}

\newcommand{\cB}{{\mathcal{B}}}
\newcommand{\cA}{{\mathcal{A}}}
\newcommand{\cD}{{\mathcal{D}}}

\newcommand{\cS}{{\mathcal{S}}}
\newcommand{\cT}{{\mathcal{T}}}

\newcommand{\Csf}{{\mathsf{C}}}
\newcommand{\Dsf}{{\mathsf{D}}}

\renewcommand{\d}{{\mathsf{d}}}

\newcommand{\ff}{{\mathbb{F}}}

\newcommand{\nn}{{\mathbb N}}

\newcommand{\bldA}{{\mbox{\boldmath $A$}}}

\newcommand{\bldb}{{\mbox{\boldmath $b$}}}

\newcommand{\blde}{{\mbox{\boldmath $e$}}}

\newcommand{\bldg}{{\mbox{\boldmath $g$}}}
\newcommand{\bldh}{{\mbox{\boldmath $h$}}}
\newcommand{\bldI}{{\mbox{\boldmath $I$}}}
\newcommand{\bldG}{{\mbox{\boldmath $G$}}}

\newcommand{\bldx}{{\mbox{\boldmath $x$}}}

\newcommand{\bldy}{{\mbox{\boldmath $y$}}}
\newcommand{\bldz}{{\mbox{\boldmath $z$}}}

\usepackage{xcolor}
\newcommand{\vit}[1]{{{\color{black} {#1}}}}

\newlength{\Algwidth}

\title{Batch Codes for Asynchronous Recovery of Data}
\author{{\bf Ago-Erik Riet, Vitaly Skachek, and Eldho K. Thomas} \\ 
Institutes of Computer Science; Mathematics and Statistics \\
Faculty of Science and Technology,
University of Tartu, Narva mnt. 18, Tartu 51009, Estonia\\
E-mail: \{ago-erik.riet, vitaly.skachek, eldho.thomas\} @ ut.ee}
\begin{document}

\maketitle

\begin{abstract}
We propose a new model of asynchronous batch codes that allow for parallel recovery of information symbols from a coded database in an asynchronous manner, i.e. when 
queries arrive at random times and they take varying time to process. We show that the graph-based batch codes studied by 
\vit{Rawat \emph{et al.}} are asynchronous. Further, we demonstrate that hypergraphs
of Berge girth larger or equal to 4, respectively larger or equal to 3, yield graph-based asynchronous batch codes,
respectively private information retrieval (PIR) codes.
We prove the hypergraph-theoretic proposition that the
maximum number of hyperedges in a hypergraph of a fixed Berge girth equals
the quantity in a certain generalization of the hypergraph-theoretic
(6,3)-problem, first posed by Brown, Erd\H os and S\'os. We then
apply the constructions and bounds by Erd\H os, Frankl and R\"odl
about this generalization of the (6,3)-problem, known as the (3$\varrho$-3,$\varrho$)-problem, 
to obtain batch code constructions and bounds on the redundancy of the graph-based asynchronous batch and PIR codes. 
We derive bounds on the optimal redundancy of several families of asynchronous batch codes with the query size $t=2$. In particular, we
show that the optimal redundancy $\rho(k)$ of graph-based asynchronous batch codes of dimension $k$ for 
$t=2$ is $2\sqrt{k}$. Moreover, for graph-based asynchronous batch codes with $t \ge 3$, 
$\rho(k) = O\left({k}^{1/(2-\epsilon)}\right)$ for any small $\epsilon>0$.

\renewcommand{\thefootnote}{\fnsymbol{footnote}}
\footnotetext{The work of Ago-Erik Riet is partially supported by the Estonian Research Council grants PSG114 and IUT20-57.
The work of Vitaly Skachek is supported in part by the Estonian Research Council grant PRG49.
The work of Eldho K. Thomas is supported in part by the 
European Regional Development Fund through Mobilitas Pluss grant MOBJD246. 
This work is also supported in part by the European Regional Development Fund via CoE project EXCITE.
Authors thank \"Ulo Reimaa for useful comments and discussions.}
\footnotetext{The material in this paper is presented in part in~\cite{Ago}.}
\end{abstract}

\begin{IEEEkeywords}
Primitive linear multiset batch codes, private information retrieval codes, extremal hypergraph theory, Tur\'an theory, packing designs. 
\end{IEEEkeywords}

\vspace{-1ex}
\section{Introduction}
Batch codes were originally proposed by Ishai \emph{et al.}~\cite{Ishai} for load balancing in distributed storage systems with multiple servers. It was also suggested in~\cite{Ishai} to use batch codes for private information retrieval. Different constructions of these codes were presented therein. 
The usefulness of batch codes for load balancing in practical distributed storage systems were further articulated 
in~\cite{Soljanin1, Soljanin2, Mikelsaar}.  
In \cite{Wang2013, WKC2015}, the authors proposed to use the so-called  “switch codes”, which is a special case of batch codes, to facilitate the routing of data in network switches (see also~\cite{Chee}). A special class of batch codes called \emph{combinatorial batch codes} was studied, for example, in \cite{Bhatt, Gal, Stinson}.

Another special class of batch codes, which is the focus of our study, is \emph{linear} (or, \emph{computational}) batch codes~\cite{Lipmaa, dimakis-batch, Zhang-Skachek, Vardy, Eldho}, where the data is viewed as elements of a finite field written as a vector, and it is encoded using a linear transformation of that vector. 

Coding schemes for private information retrieval (PIR) were proposed by Fazeli, Vardy and Yaakobi~\cite{Fazeli}. The authors showed that a family of codes called PIR codes, which is a relaxed version of batch codes, can be employed
in classical linear PIR schemes in order to reduce the redundant information stored in a distributed server system. It was suggested therein to emulate standard private information retrieval protocols using a special layer (code) which maps between the requests
of the users and the data which is actually stored in the database.  

In both of the above approaches (batch and PIR codes), typically, a distributed data storage system is considered.
The coded words are written across a block of disks (servers), where each disk stores a single symbol (or a group of symbols). The reading of data is done by accessing a small number of disks.
Mathematically, this can be equivalently represented by the assumption that each information symbol depends on a small number of other symbols. However, the type of requested queries varies in different code models. 
Specifically, in PIR codes several copies of the same information symbols are requested, while in batch codes any possible combination of different information symbols could be requested. 
That is, PIR codes of dimension $k$ support queries of the form $(\underbrace{x_i,x_i, \cdots,x_i}_t),~i \in [k],~[k] \triangleq \{1, 2, \cdots, k \}$, whereas batch codes supports queries of the form $(x_{i_1},x_{i_2},\cdots, x_{i_t})$, for possibly different indices $i_1, \cdots, i_t \in [k]$.

Linear batch codes and PIR codes have many similarities to the \emph{locally-repairable codes}~\cite{dimakis-survey},
which are used for repair of lost data in distributed data storage systems. The main difference, however, is that in the locally-repairable codes, it is 
the coded symbols that are to be repaired, while in the batch codes and PIR codes it is the information symbols that are to be reconstructed~\vit{\cite{skachek}}.  
\medskip 

In this work, we observe that (classical) batch codes start serving the user requests only after a full batch of $t$ requests have been prepared. Since the arrival time of requests is random, some requests experience a longer waiting time, which is not desirable in delay-sensitive applications. Thus, we propose a new model called \emph{asynchronous batch codes}, which is a variation of batch codes with some additional properties.

In the proposed asynchronous model, the code starts serving the requests immediately after they arrive. If at least one of the initial $t$ request is served, the code is ready to take a next request without interrupting the servers which are currently busy. Unlike regular batch codes, in the new model, one does not have to wait for a full batch of requests to arrive, and hence it can be better suited for practical purposes. However, the redundancy of asynchronous batch codes is slightly higher than that of (classical) batch codes with the same parameters, and the analysis is more difficult. 

This paper presents the first detailed study of asynchronous batch codes. We focus mainly on asynchronous batch codes that are constructed from hypergraphs, analogous to the graph-based batch codes proposed in \cite{dimakis-batch}. By using results from hypergraph theory, we derive bounds on the redundancy of such asynchronous batch codes. We also discuss properties of asynchronous batch codes, which supports smaller batch sizes and propose some explicit constructions for any batch size $t$.

The paper is organized as follows: In Section \ref{sec:notations}, we explain the notations and basic definitions. The model of asynchronous batch code as well as the graph-based model are introduced in Section \ref{sec:abc}. Some examples and basic properties are also discussed there. The connection between hypergraphs and asynchronous batch codes as well as bound computations based on the results from hypergraph theory are discussed in Sections \ref{sec:hyper}--\ref{sec:abchypergraphs}. Finally, in Sections \ref{sec:properties} and \ref{sec:general}, we consider asynchronous batch codes with batch size $t=2$ and $t>2$, respectively. Some properties and explicit constructions are given in those sections.

\section{Notation and Preliminaries}
\label{sec:notations}

\subsection{Batch and PIR Codes}

We denote by $\nn$ the set of natural numbers, and by $\ff$ a finite field.  
We use the notation $\bldI_k$ for a $k \times k$ identity matrix over $\ff$. When the value of $k$ is clear from the context, we may also use a notation $\bldI$.  
In this work, we consider only (primitive multiset) batch codes as defined in~\cite{Vardy}.
\begin{definition}[\hspace{-0.1ex}\cite{Vardy}]
\label{batch}
An $(n,k, t)$ batch code $\code$ over a finite alphabet $\Sigma$ is defined by
an encoding mapping $\Csf \; : \; \Sigma^k \rightarrow \Sigma^n$, and a decoding mapping $\Dsf \; : \; \Sigma^n \times [k]^t\rightarrow \Sigma^t$, such that  
\begin{enumerate}
\item
For any $\bldx = (x_1, x_2, \cdots, x_k) \in \Sigma^k$ and 
$i_1, i_2, \cdots, i_t \in [k] \; $, 
\[
\Dsf\left(\bldy=\Csf(\bldx), i_1, i_2, \cdots, i_t\right) = (x_{i_1}, x_{i_2}, \cdots, x_{i_t}). \; \]
\item 
The symbols in the query $(x_{i_1}, x_{i_2}, \cdots, x_{i_t})$ can be reconstructed from 
$t$ respective pairwise disjoint recovery sets of symbols of $\bldy = (y_1, y_2, \cdots, y_n) \in \Sigma^n$ (the symbol $x_{i_\ell}$ is reconstructed from the $\ell$-th recovery set for each $\ell$, $1 \le \ell \le t$). 
\end{enumerate}
\label{def:batch}
\end{definition}

\begin{definition}
A recovery set of size one is called a \emph{singleton}. 
\end{definition}

Let $\ff = \ff_q$ be a finite field with $q$ elements, where $q$ is a prime power,
and $\code$ be a linear $[n,k]$ code over $\ff$. Denote the redundancy $\rho \triangleq n-k$. 
\begin{definition}
\label{def:linear_batch}
A \emph{linear batch code} is a batch code where the encoding of $\Csf$ is given as a multiplication 
by a $k \times n$ generator matrix $\bldG$ over $\ff$ of an information vector $\bldx \in \ff^k$, 
\begin{equation}
\bldy = \bldx \cdot \bldG \; ;~~ \bldy \in \ff^n .
\label{def:linear}
\end{equation}
A \emph{linear batch code} with the parameters $n$, $k$ and $t$ over $\ff_q$, where $t$ is a number of queried symbols, is denoted as an \emph{$[n,k,t]_q$-batch code}.
Sometimes we simply write $[n,k,t]$-batch code if the value of $q$ is clear from the context.  
\end{definition}

\begin{definition}
\label{def:query_size}
An \emph{$[n,k,t,r]_q$-batch code} (or, simply, \emph{$[n,k,t,r]$-batch code}) is an $[n,k,t]_q$-batch code ($[n,k,t]$-batch code, respectively) such that 
the size of every recovery set is less or equal to $r$.
\end{definition}

\begin{definition}
\label{def:systematic_batch}
A linear batch code is called \emph{systematic} if the matrix $\bldG$ has the form $[\bldI_k | \bldA]$, where $\bldA$ is a $k \times \rho$ matrix over $\ff$.
\end{definition}
For a systematic code $\code$, the encoding takes the form 
\[
  \bldy = \bldx \cdot \bldG = ( \; \bldx \; | \; \bldz \; ) \; , \quad \mbox{ where } \bldz = \bldx \cdot \bldA \; . 
\]
The subvector $\bldx$  of $\bldy$ is called the \emph{systematic part of $\bldy$}, and its symbols are called \emph{information symbols}. 
The subvector $\bldz$ is called a \emph{redundancy part  of $\bldy$}, and its symbols are called \emph{parity symbols}. 
Similarly, the submatrix $\bldI_k$ of $\bldG$ is called the systematic part of $\bldG$, and the submatrix $\bldA$ of $\bldG$ is called a redundancy part of $\bldG$.

\begin{definition}[\hspace{-0.1ex}\cite{Fazeli}]
\emph{Linear PIR codes} are defined similarly to linear primitive multiset batch codes, with a difference that the supported queries are of the form $(x_i, x_i, \cdots, x_i), \; i \in [k],$
(and not $(x_{i_1}, x_{i_2}, \cdots, x_{i_t}), \; i_1, i_2, \cdots, i_t \in [k]$ as in batch codes). 
\end{definition}

For constructions of PIR codes see, for example,~\cite{Lin, Vajha}. In what follows, we consider linear batch codes and PIR codes over $\ff = \ff_2$, yet most of the results hold for codes over larger fields too.  

\subsection{Graphs and Hypergraphs}
\label{gh}
Let $W^{(\varrho)}$, $\varrho \ge 2$, denote the set of all unordered $\varrho$-tuples of distinct elements of the set $W$.
An (undirected) \emph{graph} $G(V,E)$ consists of a finite set $V$, called the \emph{vertex set} and a finite set $E\subseteq V^{(2)}$ of pairs of vertices, called the \emph{edge set}. The graph $G(V,E)$ is \emph{bipartite} with \emph{bipartition} (or \emph{parts}) $(A,B)$ if $A\cup B=V$, $A\cap B = \varnothing$, and $|A\cap e|=1$ and $|B\cap e| = 1$ for every edge $e\in E$. We denote the bipartite graph with distinguished parts $A$ and $B$ as $G(A,B,E)$ where we call $A$ the \emph{left part} and $B$ the \emph{right part}. A \emph{$b$-cycle} in a graph $G(V,E)$ is a cyclic sequence of $b$ vertices and $b$ edges, alternatingly between vertices and edges, such that each edge consists precisely of the two vertices on each side of it in the sequence. A bipartite graph $G(A,B,E)$ is \emph{left-regular} if all \emph{left degrees} $\d (a)\triangleq |\{e\in E\, : \, a\in e\} |$, where $a\in A$, are equal.

More generally, a hypergraph $\mathcal{G}(V,E)$ consists of a finite set $V$ of vertices and a finite collection $E$ of subsets of $V$, called (hyper)edges. The hypergraph is \emph{$\varrho$-uniform}, or an \emph{$\varrho$-graph}, if each edge consists of the same number $\varrho$ of vertices, that is, $E \subseteq V^{(\varrho)}$. Thus, a graph can be viewed as a $2$-uniform hypergraph. 

A \emph{Berge cycle} in a hypergraph is a sequence $(e_1,v_1,e_2,v_2,\ldots,v_b,e_{b+1})$ where $e_1,e_2,\ldots,e_b$ are distinct hyperedges, $v_1,v_2,\ldots,v_b$ are distinct vertices, $v_{i-1}, v_i\in e_i$ for all $i$ (we have taken all indices modulo $b$ when defining the sequence) and $e_1=e_{b+1}$. We define a \emph{Berge path} in a hypergraph similarly to be an alternating sequence of vertices and hyperedges \vit{that starts and ends with a vertex}, where all hyperedges are distinct and all vertices are distinct, 
and the vertices on each side of a hyperedge in the sequence belong to the hyperedge. A hypergraph is \emph{Berge-connected} if there is a Berge path from any vertex to any other vertex. Equivalently, a hypergraph is \emph{Berge-disconnected} if its vertex set $V$ can be partitioned into two non-empty sets $V=V_1\cup V_2$ such that, for each hyperedge $e$, either $e\cap V_1=\varnothing$ or $e\cap V_2=\varnothing$; it is \emph{Berge-connected} if it is not disconnected. A hypergraph is said have \emph{Berge girth} equal $k$ if (a) it contains a Berge cycle with $k$ hyperedges; (b) it contains no Berge cycles with fewer than $k$ hyperedges. If a subset of vertices is allowed several (a finite number of) times as a hyperedge, we have a multihypergraph.
We note that a multi-$\varrho$-graph for $\varrho\geq 2$ with Berge girth at least 3 is necessarily a simple hypergraph, i.e.~no subset of vertices appears as an edge several times.

The following definition of the correspondence between bipartite graphs and (multi)hypergraphs is instrumental for the analysis in this paper. 
\begin{definition}
\label{def:correspondence}
With a (multi)hypergraph $\mathcal{G}(V,E)$ one can associate the bipartite \emph{incidence graph} $G(E,V,I)$ with left part $E$ and right part $V$ where $\{e,v\}$ is an edge, i.e.~$\{e, v\}\in I$ in $G$, if and only if $v\in e$ in $\mathcal{G}$. By going backwards, given a bipartite graph $G(E,V,I)$ we construct a (multi)hypergraph $\mathcal{G}(V,E)$ by identifying each $e\in E$ with the set $\{v\in V\,|\,\{e,v\}\in I\}$. 
\end{definition}
Therefore, multihypergraphs are in one-to-one correspondence with bipartite graphs. 
A multihypergraph is Berge-connected if and only if its incidence graph is connected; there is a one-to-one correspondence between Berge cycles with $k$ hyperedges in the multihypergraph and cycles of length $2k$ in the incidence graph. 

Extremal graph theory (or Tur\'an theory) is the study of maximal graphs with some properties. A typical question is to find the maximum size of a graph (number of hyperedges) on $n$ vertices, provided it contains no copy of a fixed subgraph, such as the triangle. 

A $\varrho$-graph $\mathcal{G}'(V',E')$ is a \emph{sub-$\varrho$-graph} of a $\varrho$-graph $\mathcal{G}(V,E)$ if $V'\subseteq V$ and $E'\subseteq \{e\in E \,|\, e\subseteq V'\}$. We say that the sub-$\varrho$-graph is \emph{induced} by the vertex set $V'$ if in addition $E'= \{e\in E \,|\, e\subseteq V'\}$. Similarly we say that a subset of hyperedges $E'$ \emph{induces} the vertex set $\bigcup_{e\in E'} e$. 

\subsection{Graph-based Batch and PIR Codes}

\begin{construction}
\label{constr:bipartite}
Let $\code$ be an $[n, k, t]_q$ batch (PIR) code defined by a systematic encoding matrix $\bldG = \left[ \; \bldI \; | \; \bldA \; \right]$.
Assume that all the used recovery sets of $\code$ contain a single parity symbol and any number of information symbols. Note that, while other recovery sets may exist, we only allow the system to use such recovery sets. 
The following bipartite graph representation of $\code$ was proposed 
in~\cite{dimakis-batch}.
 
Let $G(A,B,E)$ be a bipartite graph, where $A$ is the set of the information symbols,
$B$ is the set of the parity symbols, and  
\[
E = \Big\{ \{ u, v \} : u \in A, v\in B, \mbox{ information symbol } u \mbox{ participates in parity symbol } v \Big\} \; . 
\]
\end{construction}

\begin{definition}
An asynchronous $[n,k,t]$-batch code, which can be represented as in Construction~\ref{constr:bipartite}, is called a {\bf graph-based asynchronous} $[n,k,t]$-batch code. 
\end{definition}

\begin{theorem}(\hspace{-0.1ex}\cite[Theorem 1 and Lemma 2]{dimakis-batch})
\label{theorem1-dimakis}
Let $\code$ be an $[n, k]$ systematic code represented by the bipartite graph $G(A, B, E)$. 
Assume that there exists an \emph{induced subgraph} $H(A, B', E')$ of $G$, that is, $B' \subseteq B$ and $E' = \{e\in E \, :\, |e\cap B'|=1 \}$, such that: 
\begin{itemize}
\item[(i)] Each vertex in $A$ has degree at least $t$ in the bipartite graph $H$.
\item[(ii)] The graph $H$ has girth $\ge 8$ (respectively, $\ge 6$).
\end{itemize}
Then, $\code$ is an $[n, k, t]$ batch code (respectively, PIR code).
\end{theorem}


It follows from Theorem~\ref{theorem1-dimakis} that constructions of left-regular bipartite graphs  without short cycles yield constructions of batch and PIR codes.
In what follows, we use this approach in order to construct \vit{batch and PIR codes with good parameters}. 
Specifically, we use known constructions of good hypergraphs, which can be 
mapped to bipartite graphs without short cycles, in order to construct good codes. 

\section{Asynchronous batch codes}
\label{sec:abc}
In this section, we introduce a new special family of batch codes, termed \emph{asynchronous batch codes}. 
Assume that $\code$ is a linear $[n, k, t]$ batch code over $\ff$ as in Definition \ref{def:linear_batch}, used for retrieving a batch of $t$ symbols $(x_{\ell_1}, x_{\ell_2}, \cdots, x_{\ell_t})$, $\ell_{i} \in [k]$, $i \in t$, in parallel from a coded database that consists of $n$ servers, such that at most one symbol is retrieved from each server. To this end assume that the queries arrive at random times, and that the response time of the servers for different requests varies, and thus some symbol $x_{\ell_j}$ (w.l.o.g.) can be retrieved faster than the other symbols. In asynchronous retrieval mode, once $x_{\ell_j}$ was retrieved, it is possible to
retrieve any other request $x_{\ell_{t+1}}$, $\ell_{t+1} \in [k]$, in parallel to retrieving of $(x_{\ell_1}, x_{\ell_2}, \cdots, x_{\ell_{j-1}}, x_{\ell_{j+1}}, \cdots, x_{\ell_t})$, without reading more than one symbol from each server, and without changing (interfering with) the servers in use in retrieving the other symbols. In that way, the asynchronous batch codes support (asynchronous) retrieval of $t$ symbols in parallel. We proceed with a formal definition. 

\begin{definition}
An asynchronous (linear primitive multiset) $[n, k, t]$-batch code $\code$ is a (linear primitive multiset) batch code with the additional property that for any legal batch of queries $(x_{\ell_1}, x_{\ell_2}, \cdots, x_{\ell_t})$, for any $j \in [t]$, 
it is always possible to replace $x_{\ell_j}$ by some $x_{\ell_{t+1}}$, $\ell_{t+1} \in [k]$, such that 
$x_{\ell_{t+1}}$ is retrieved from the servers not in use for the retrieval of $x_{\ell_1}, x_{\ell_2}, \cdots, x_{\ell_{j-1}}, x_{\ell_{j+1}}, \cdots, x_{\ell_t}$, without reading more than one symbol from each server. Note that a component of this definition is an algorithm that can specify at each step which available recovery set will be used.
\end{definition}

\begin{example}
Consider the systematic $[8,4,3]_2$ batch code $\code$ generated by the matrix
\[
\bldG \; = \; \left(
\begin{matrix}
1&0&0&0&1&0&1&0 \\
0&1&0&0&1&0&0&1 \\
0&0&1&0&0&1&1&0  \\
0&0&0&1&0&1&0&1
\end{matrix}
\right) \; .
\]
There are three disjoint recovery sets for each $x_i$, $i \in [4]$: 
\begin{itemize}
\item $x_1 = y_1, \; x_1 = y_2 + y_5, \; x_1 =y_3 + y_7\;; $
\item $x_2 = y_2, \; x_2 = y_1 + y_5, \; x_2 =y_4 + y_8\;; $
\item $x_3 = y_3, \; x_3 = y_4 + y_6, \; x_3 =y_1 + y_7\;; $
\item $x_4 = y_4, \; x_4 = y_3 + y_6, \; x_4 =y_2 + y_8\;. $
\end{itemize}

We verify that $\code$ is an asynchronous batch code which supports any $t=2$ requests. We observe that irrespectively of the first query $x_i$ and the corresponding recovery set which is being served, the system is able to serve any additional query $x_j$, $j \in [4]$.
Indeed, any recovery set uses at most two different variables $y_{\ell_1}$ and $y_{\ell_2}$, $\ell_1, \ell_2 \in [8]$. Since the new query $x_j$ has three possible different recovery sets,  at least one of these sets contains neither $y_{\ell_1}$ nor $y_{\ell_2}$, and therefore it can be used without 
using the same server more than once. 

We conclude that $\code$ is an asynchronous $[8,4,2]$-batch code.
On the other hand, $\code$ is not an asynchronous $[8,4,3]$-batch code, since if the pair of requests $(x_1, x_2)$ is being served using the recovery sets $x_1=y_3+y_7$ and $x_2=y_1+y_5$, respectively, then an additional request $x_1$ can not be served without using the same server more than once.
\end{example}

There is a conceptually simple but computationally expensive (and thus difficult to use) necessary and sufficient condition to check if a given $k \times n$ generator matrix gives an asynchronous $[n,k,t]$ batch code over $\ff$. It is as follows. 

Consider the hypergraph $H(V,E)$ whose vertices are the columns of the generator matrix, and whose hyperedges are the (containment-wise) minimal subsets of columns the elements of which modulo 2 sum to a unit vector, i.e.~the hyperedges correspond to the recovery sets of an information symbol. Give each hyperedge a label $\ell$, $\ell \in [k]$. This label denotes the information symbol, which the edge recovers. Then the generator matrix generates an asynchronous $[n,k,t]$ batch code if and only if
 there exists a hypergraph $H'(E,F)$ such that:
\begin{enumerate}
    \item $\varnothing \in F$ \newline ($F$ is \emph{non-empty}),
    \item For any $f\in F$ with $|f|<t$, and for any $i \in [k]$, there exists $e\in E$ with label $i$ such that $f\cup\{e\}\in F$ \newline (\emph{extension property}),
    \item For any $f\in F$ and and for any $f'\subseteq f$ we have $f'\in F$ \newline (\emph{hereditary property}),
    \item For any $f\in F$ and $e,e'\in f$ with $e\not=e'$ we have $e\cap e'=\varnothing$ (\emph{pairwise disjointness}).
\end{enumerate}

The intuition is that the collection $F$ consists of sets of recovery sets that are allowed to be used simultaneously in the system. Once a suitable $H'=(E,F)$ is found, this information can be provided to (hard-coded in) the system that is using this batch code, in order to facilitate its operation and guarantee correctness.

This description as a hypergraph provides a way to algorithmically check whether a given binary matrix gives an asynchronous batch code that supports any $t$ queries. The matrix has not be systematic in order to use the algorithm. However, the computational complexity increases very fast as $t$ increases.

It is straightforward to see that any asynchronous $[n, k, t]_q$ batch code is an $[n, k, t]_q$ batch code. The opposite, however, does not always hold. 

\begin{example}
Consider batch codes, which are obtained by taking simplex codes as suggested in \cite{WKC2015}. 
The $[7,3,4]$-batch code $\code$ is formed, for example, by the generator matrix 
\[
\bldG = \left(
\begin{matrix}
1&0&0&1&1&0&1 \\
0&1&0&1&0&1&1 \\
0&0&1&0&1&1&1
\end{matrix}
\right)
\]
is a $[7,3,4]_2$ batch code. Assume that the query $(x_1, x_1, x_1, x_1)$ was submitted by the users. Then, one copy of $x_1$
is retrieved from $y_1$, and for each of the remaining three copies of $x_1$, at least two symbols of $\bldy$ have to be used. 
Assume that the query that uses $y_1$ has been served, but the remaining queries are still being served. 
If the next query $x_2$ arrives, it is impossible to serve it without accessing one of the servers containing 
$y_2, \cdots, y_7$ at least twice. Therefore, $\code$ is not an asynchronous $[7,3,4]_2$ batch code. 
\end{example}

We denote by $\cA(k,t)$ and $\cB(k,t)$ the minimal length $n$ of the asynchronous and general linear $[n,k,t]$-batch codes, respectively.   
Similarly, by $\cA(k,t,r)$ we denote the minimal length $n$ of the asynchronous $[n,k,t,r]$-batch codes.

\begin{example}
Consider the case $k =2$ and $t=2$. 
It is known that $\cB(2,2) = 3$, and the only generator matrix $\bldG$ for such a code (up to a permutation of columns) is:  
\[
\bldG = \left( \begin{array}{ccc}
1 & 0 & 1 \\
0 & 1 & 1 \\
\end{array} \right) \; 
\] 
(it is straightforward to see that any binary $2 \times 3$ matrix $\bldG$ with repeating columns does not produce a batch code with $t = 2$). 

This $\bldG$ does not correspond to an asynchronous batch code with $t=2$. To see that, consider the sequence of requests $x_1, x_1, x_2$. 
The first $t=2$ request are recovered as $x_1 = y_1$ and $x_1 = y_2 + y_3$. Assume that $x_1 = y_1$ was served first, and now try to assign a recovery 
set for $x_2$. It is impossible. We conclude that $\cA(2,2) \ge 4 > 3 = \cB(2,2)$. 

On the other hand, the code with
\[
\bldG = \left( \begin{array}{cccc}
1 & 1 & 0 & 0  \\
0 & 0 & 1 & 1 \\
\end{array} \right) \; 
\] 
is an asynchronous batch code, and therefore $\cA(2,2) = 4$. 
\end{example}

\begin{lemma}(\hspace{-0.1ex}\cite[Lemma 3]{dimakis-batch})
\label{lemma3-dimakis}
Let $\code$ be an $[n, k]$ systematic code represented by the bipartite graph $G(A, B, E)$. 
Assume that there exists an \emph{induced subgraph} $H(A, B', E')$ of $G$, that is, $B' \subseteq B$ and $E' = \{e\in E \, :\, |e\cap B'|=1 \}$, such that: 
\begin{itemize}
\item[(i)] Each vertex in $A$ has degree at least $t$ in the bipartite graph $H$.
\item[(ii)] The graph $H$ has girth at least 8.
\end{itemize}
Then, each message symbol has at least $t$ disjoint recovery sets. Moreover, for any $i,j \in [k]$, $i \neq j$, any one of the disjoint recovery sets for the message symbol $x_i$ has common symbols 
with at most one of the disjoint recovery sets for the message symbol $x_j$. 
\end{lemma}
It turns out, that the conditions in Theorem~\ref{theorem1-dimakis} yield asynchronous batch codes. More formally:

\begin{theorem}
\label{thrm-asynchronous}
Let $\code$ be an $[n, k]$ systematic code represented by the bipartite graph $G(A, B, E)$. 
Assume that there exists an \emph{induced subgraph} $H(A, B', E')$ of $G$, that is, $B' \subseteq B$ and $E' = \{e\in E \, :\, |e\cap B'|=1 \}$, such that: 
\begin{itemize}
\item[(i)] Each vertex in $A$ has degree at least $t$ in the bipartite graph $H$.
\item[(ii)] The graph $H$ has girth at least 8.
\end{itemize}
Then, $\code$ is an \underline{asynchronous} $[n, k, t]$ batch code.
\end{theorem}
 
\begin{proof}
Assume that $\code$ is an $[n, k]$ systematic code satisfying the conditions in the theorem. 
We prove that $\code$ is an asynchronous $[n,k,t]$ batch code. From Theorem \ref{theorem1-dimakis}, $\code$ is an $[n, k, t]$ batch code, and therefore 
there exist $t$ disjoint recovery sets for each of the information symbols $x_i$, $i \in [k]$. 

Assume that a set of $t-1$ queries $x_{i_1}, x_{i_2}, \ldots, x_{i_{t-1}}$ is currently being served using the recovery sets $S_{i_1}, S_{i_2}, \cdots,$ $S_{i_{t-1}}$, respectively (if the number of queries being served is smaller than $t-1$, then exactly the same proof applies). Let $x_j$ be the new query. We show that (irrespectively of the recovery sets used for the above $t-1$ queries) there exists a recovery set for $x_j$, which is pairwise disjoint with each of these $t-1$ recovery sets. 

Consider $t$ recovery sets for $x_j$: $T_1, T_2, \cdots, T_t$. Due to Lemma~\ref{lemma3-dimakis}, each of the sets $S_{i_j}$, $j \in [t-1]$, overlaps with at most one of the sets $T_\ell$, $\ell \in [t]$. Therefore, due to the pigeon-hall principle, there exists a set $T_{\ell'}$, ${\ell'} \in [t]$, which is pairwise disjoint with all of $S_{i_1}, S_{i_2}, \cdots, S_{i_{t-1}}$, and it can be used to recover $x_j$. We conclude that $\code$ is an $[n,k,t]$ asynchronous batch code. 
\end{proof}

The above theorem shows that the graph-based batch codes given as in Theorem ~\ref{theorem1-dimakis} are asynchronous with the same parameters.
We also present the following converse result. 

\begin{proposition}
Let $\code$ be a systematic asynchronous $[n, k, t]$-batch code generated by the matrix $\bldG$ with the least number of ones, represented by the bipartite graph $G(A, B, E)$. Additionally, let each recovery set be either a singleton or a single column in the redundancy part together with columns in the systematic part. 
Assume that there exists an \emph{induced subgraph} $H(A, B', E')$ of $G$, that is, $B' \subseteq B$ and $E' = \{e\in E \, :\, |e\cap B'|=1 \}$, such that
for each vertex $a_i\in A$ there exists at least one neighboring vertex $b' \in B'$ with degree $\ge 2$ in $H$.
Then, each vertex in $A$ has degree at least $t$ in $H$.
\end{proposition}
\begin{proof}
Assume to the contrary that a vertex $a_j \in A$ has degree $t-1$ in $G$. Then, row $j$ in $\bldG$ has weight $t$. By condition (i), there exists at least one column $\bldb'$ in the redundancy part of $\bldG$ with weight larger or equal to 2, $\bldb'$ has ones in rows $j$ and $i$, $i \neq j$.

Next, assume that the information symbol $x_i$ (corresponding to the vertex $a_i \in A$) is currently being recovered using the (non-singleton) recovery set $R_{a_i}$, which includes the column $\bldb'$ and the singleton column with one in row~$j$. 
Observe that if the column $\bldb'$ is never used to recover $x_i$, we can replace one by zero in position $i$ of $\bldb'$. The resulting batch code has the same parameters, but has a smaller number of ones in the generator matrix $\bldG$, thus yielding a contradiction to the minimality.

Notice that the two columns with ones in position $j$ (the singleton column and $\bldb'$) are currently busy recovering $x_i$. Therefore, if the additional $t-1$ queries are $(x_j, x_j, \cdots, x_j)$, then there are not enough recovery sets for $x_j$ available to serve all those requests. This is in contradiction to the fact that $\code$ is an asynchronous $[n,k,t]$ batch code.
\end{proof}




\section{Hypergraph Theory}
\label{sec:hyper}
In~\cite{brown1}, \cite{brown2}, Brown, Erd\H os and S\'os pose the following extremal combinatorial problems on $\varrho$-graphs. Let $f^{(\varrho)}(n;\kappa,s)$ denote the smallest $t$ such that every $\varrho$-graph on $n$ vertices with $t$ hyperedges contains at least one sub-$\varrho$-graph on $\kappa$ vertices with $s$ hyperedges. Therefore $f^{(\varrho)}(n;\kappa,s)-1$ is the maximum size of an $\varrho$-graph whose no set of $\kappa$ vertices contain $s$ or more hyperedges. The authors are interested in bounds on this quantity for fixed $\varrho$, $\kappa$ and $s$. The resolution of the first interesting open case $f^{(3)}(n;6,3)$, known as the $(6,3)$-problem, by Ruzsa and Szemer\'edi~\cite{ruzsa} is a classical result in extremal combinatorics. Erd\H os, Frankl and R\"odl~\cite{erdos} extended this result to any fixed $\varrho$, also giving an easier construction for the lower bound, thus solving the so-called $(3\varrho-3,3)$-problem. There are various later generalizations of~\cite{ruzsa} and~\cite{erdos}, see for example~\cite{alon} and the references therein, and the survey~\cite{furedi}.

In what follows, we show that finding the maximum size (the number of hyperedges) of a hypergraph with a given Berge girth is essentially a generalization of the $(6,3)$-problem for 3-graphs which would be called the $(\kappa\varrho-\kappa,\kappa)$ problem for $\varrho$-graphs in this terminology. Finally, we apply the resolution of the $(3\varrho-3,3)$ problem in~\cite{erdos} to batch codes.

\begin{theorem}\label{berge girth}
Let $G^{(\varrho)}(n,\kappa)$ be the maximum size of an $\varrho$-graph with $n = |V|$ containing no Berge cycle of length $\kappa$ or less (i.e.~of Berge girth at least $\kappa+1$). Let $F^{(\varrho)}(n;h,s)=f^{(\varrho)}(n;h,s)-1$ be the maximum size of an $\varrho$-graph, no $h$ of whose vertices contain $s$ or more hyperedges. Then $F^{(\varrho)}(n;\kappa\varrho-\kappa,\kappa) = G^{(\varrho)}(n,\kappa)$.
\end{theorem}
We prove this theorem in more generality by using the following lemmas.

\begin{lemma}\label{help}
\label{lemma:berge}
For a Berge-connected hypergraph $\mathcal{G}(V,E)$ with $|V|\geq 2$ we have: 

\begin{enumerate}
\item $\sum_{e\in E} (|e|-1) \geq |V|-1.$
\item $\mathcal{G}(V,E)$ contains no Berge cycles (is a Berge tree) if and only if
$\sum_{e\in E} (|e|-1) = |V|-1$. 
\item
$\mathcal{G}(V,E)$ contains exactly one Berge cycle if and only if $\sum_{e\in E} (|e|-1) = |V|.$
\end{enumerate}
\end{lemma}

\begin{proof}
Consider the bipartite incidence graph $G(E,V,I)$ of $\mathcal{G}(V,E)$ where $I \triangleq \{\{e,v\}\, | \, v\in e, \, v\in V, \, e \in E\}$. It is a connected graph with $|V|+|E|$ vertices and at least $|V|+|E|-1$ edges.
The hypergraph $\mathcal{G}(V,E)$ does not have cycles if and only if so does $G(E,V,I)$. The hypergraph $\mathcal{G}(V,E)$ has one Berge cycle if and only if  $G(E,V,I)$ has one cycle. 

If $G(E,V,I)$ has no cycle, then there exists an ordering on $E = \{ e_1, e_2, \cdots, e_{|E|}\}$, such that $e_1$ is incident with $\varrho$ vertices in $V$, $e_2$ is incident with $\varrho-1$ additional vertices in $V$, and $e_{|E|}$ is incident with $\varrho-1$ additional vertices in $V$. By a simple counting argument we obtain Condition 1. 
Conditions 2 and 3 follow from the standard graph-theoretic arguments about spanning trees with an added edge.
\end{proof}

Next, we introduce the following definition. 
\begin{definition}
\label{def:condition}
A hypergraph \emph{satisfies the condition of the $(\kappa \varrho-\kappa,\kappa)$ problem (or $(\kappa \varrho-\kappa,\kappa)$-condition, in short)} if no set of $\kappa \varrho-\kappa$ of its vertices contains $\kappa$ or more hyperedges.
\end{definition}

As we show in the sequel, an $\varrho$-graph satisfying the $(\kappa \varrho-\kappa,\kappa)$-condition can be modified to additionally have Berge girth at least $\kappa+1$ while keeping the same number of hyperedges; an $\varrho$-graph of Berge girth at least $\kappa+1$ already satisfies the $(\kappa \varrho-\kappa,\kappa)$-condition.

\begin{lemma}
An $\varrho$-graph of Berge girth at least $\kappa+1$ satisfies the $(\kappa \varrho-\kappa,\kappa)$-condition.
\label{lemma:condition}
\end{lemma}

\begin{proof}
Consider any $\kappa$ hyperedges of this graph. They do not induce any Berge cycles. 

For each of the Berge-connected components (maximal connected subhypergraphs) $\graph'(V',E')$ of the hypergraph induced by these $\kappa$ hyperedges, we have $\sum_{e\in E'} (|e|-1) = |V'|-1$ by Condition 2 of Lemma~\ref{help}, and therefore $\sum_{e\in E} (|e|-1) = \kappa(\varrho-1) = |V|-c$ for the hypergraph induced by these $\kappa$ hyperedges, where $c\geq 1$ is the number of Berge-connected components. Hence the number of vertices induced by these $\kappa$ hyperedges is $\kappa(\varrho-1)+c > \kappa \varrho-\kappa$. Thus the hypergraph satisfies the $(\kappa \varrho-\kappa,\kappa)$-condition.
\end{proof}



\begin{lemma}
An $\varrho$-graph that satisfies the $(\kappa\varrho-\kappa,\kappa)$-condition can be changed (its hyperedges can be re-wired) so that it still has the same number of hyperedges, still satisfies the $(\kappa\varrho-\kappa,\kappa)$-condition, and has Berge girth at least $\kappa+1$.
\label{lemma:rewire}
\end{lemma}

\begin{proof}
If an $\varrho$-graph satisfies the $(\kappa\varrho-\kappa,\kappa)$-condition, then from Definition~\ref{def:condition}, the total number of vertices used by any $\kappa$ hyperedges is at least $\kappa(\varrho-1)+1$. 

We consider two cases. 
\begin{description}
\item[\it Case 1: the graph induced by these hyperedges is connected.] $\,$ \newline
In this case, the graph would not contain any Berge cycles by Lemma~\ref{help}. Therefore, there is an ordering of the $\kappa$ hyperedges where the first edge would use $\varrho$ new vertices and each of the next edges would use $\varrho-1$ vertices not used so far. 
Thus, there is no Berge cycle on $\le \kappa$ hyperedges. 
\item[\it Case 2: the induced graph is disconnected.] $\,$ \newline
Consider a Berge-connected component which has some small Berge cycles, i.e.~the sub-hypergraph $\graph'(E',V')$ induced by the vertices of this component satisfies $\sum_{e\in E'} (|e|-1) > |V'|-1$, and there is a cycle of $\kappa$ or fewer hyperedges. This component has fewer than $\kappa$ hyperedges. Otherwise, the hyperedges of a small cycle together with possibly some other hyperedges forming a connected component of $\kappa$ hyperedges violate the $(\kappa \varrho -\kappa,\kappa)$-condition. This is because there would have to be an ordering of $\kappa$ of the hyperedges (which span a connected subhypergraph) where the first edge uses $\varrho$ new vertices, and each next edge uses at most $\varrho-1$ new vertices, with one of them (the last hyperedge of a cycle) using strictly fewer than $\varrho-1$ new vertices. Thus, indeed, these $\kappa$ hyperedges would violate the $(\kappa \varrho -\kappa,\kappa)$-condition. Take any hyperedge $e$ of a cycle of $\kappa$ or fewer hyperedges and re-wire it by deleting one of its vertices and adding in another vertex: there are at least two vertices of this hyperedge shared with the union of the other hyperedges of the cycle. Delete one of them from $e$ and add into $e$ a vertex from outside the connected component.
The procedure either strictly reduces the number of connected components of size smaller than $\kappa$ hyperedges, or strictly increases the total number of vertices in such components, therefore we can only repeat it a finite number of times, and eventually, when it can not be repeated anymore we will have no Berge cycles with $\kappa$ or fewer hyperedges (see Lemma~\ref{help}).
\end{description}
\end{proof}

Lemmas~\ref{help}--\ref{lemma:rewire} imply Theorem~\ref{berge girth}.

\section{PIR codes from hypergraphs of Berge girth at least 3}
 
In what follows, we remark that optimal hypergraphs of Berge girth at least 3 can be used in constructing PIR codes. 
\begin{definition}
A $\tau-(\eta,\varrho,\lambda)$ packing design is an $\varrho$-graph consisting of $\eta$ vertices (called points) and of edges (called blocks) such that each $\tau$-tuple of vertices (points) is contained in at most $\lambda$ edges (blocks). 
\end{definition}

Consider an $\varrho$-graph $\graph(V,E)$, where $V$ is a point set and $E$ is a block set, $|V|=\eta$. 
If $\graph(V,E)$ has Berge girth at least 3, then it is also a so-called \emph{$2-(\eta,\varrho,1)$ packing design}. 
When each pair of points is contained in a unique block, we have a \emph{Steiner 2-design}, also known as a combinatorial \emph{$2-(\eta,\varrho,1)$ block design}, or a \emph{$(\eta,V,1)$-BIBD (balanced incomplete block design)}. 

The maximum size $D(\eta,\varrho)$ of a packing design (i.e.~the maximum number of blocks in it) is bounded from above by the well-known improved 1st and 2nd Johnson bounds~\cite{johnson2}, see also~\cite{mills}. 
It follows from the result on the existence of designs~\cite{keevash}, that for all sufficiently large $\eta$, there is a packing design attaining either the improved 1st or 2nd Johnson bound, see also~\cite{horsley} referring to an earlier version of~\cite{keevash}. 

This means that, for large enough $\eta$, PIR codes constructed using packing designs following 
the ideas of Theorem~\ref{theorem1-dimakis} always exist, as it is stated in the following theorem.

\begin{theorem}
It holds
\begin{equation}
\lim_{\eta\rightarrow\infty} \frac{D(\eta,\varrho)}{{\eta\choose 2} / {\varrho\choose 2}} = 1 \; .
\label{eq:steiner}
\end{equation}
\end{theorem}

The existence of Steiner 2-designs for all admissible large enough $\eta$, i.e.~large enough $\eta$ satisfying some simple necessary divisibility conditions, follows from earlier works of Wilson~\cite{wilson1}, \cite{wilson2}, \cite{wilson3}. In these works, however, there is no attempt to understand the size of the lower bound on such $\eta$.

We obtain constructions of families of PIR codes with $k$ information symbols, whose redundancy $\rho = \rho(k)$ is close to a solution of the equation:  
\[
{\rho \choose 2} \Big/ {\varrho \choose 2} = k \; . 
\]
In particular, we obtain $\rho = \Theta(\sqrt{k})$ for $t = 4, 5, 6$ (equivalently, for $\varrho = 3, 4 ,5$).  
For $t \ge 7$  (equivalently, for $\varrho \ge 6$), the existence of the optimal PIR codes follows from 
the existence of the optimal Steiner-2-designs. 

To this end, we note that Fazeli, Vardy and Yaakobi in~\cite{Fazeli} also use Steiner 2-designs to construct PIR codes. In this work, however, we 
obtain similar results by using Theorem~\ref{theorem1-dimakis}.  


\section{Batch codes from hypergraphs of Berge girth at least 4}
\label{sec:abchypergraphs}
Bounds and constructions for $\varrho$-graphs $\mathcal{G}(V,E)$ with $n$ vertices of Berge girth at least 4 can be given via the $(3\varrho-3,3)$-problem in the language of the $(6,3)$-problem, as seen from Theorem~\ref{berge girth} and Lemmas~\ref{lemma:berge}. Bounds apply directly, while constructions may need to be modified slightly to lose small Berge cycles. Erd\H os, Frankl and R\"odl~\cite{erdos} address precisely the $(3 \varrho -3,3)$-problem. The authors modify the celebrated construction of Behrend~\cite{behrend} of a large subset of  $\{1,\ldots,N\}$ which contains no 3-term arithmetic progression (\emph{3AP-free}), that is, no three distinct numbers of the form $a$, $a+b$ and $a+2b$. 
This way, the authors of~\cite{erdos} construct $\varrho$-graphs with the number of hyperdges asymptotically larger than $n^{2-c}$ for any $c>0$. The construction produces a hypergraph of Berge girth at least 4, so there is no need to modify the construction. The authors also prove an upper bound $o(n^2)$ on the maximum number of hyperedges, using an early version of the Szemer\'edi's Regularity Lemma, see for example~\cite{komlos} and~\cite{skokan}.

For the sake of completeness, we recall the construction in~\cite{erdos} for large $\varrho$-graphs of Berge girth at least 4. This construction gives rise to primitive multiset linear batch codes. The authors of~\cite{erdos} prove the following Lemma, with the proof closely related to Behrend's original construction of large 3AP-free sets in~\cite{behrend}.

\begin{lemma}
\label{lemma-3PA}
There exists a set of positive integers $A\subseteq \{1,2,\ldots, n\}$ not containing three terms of any arithmetic progression of length $\varrho$, such that $$|A|\geq \frac{n}{e^{c\log \varrho \sqrt{\log n}}}$$ for some absolute constant $c>0$.
\end{lemma}

\begin{proof}
Omitted. Please see~\cite{erdos} and~\cite{behrend} for more details.
\end{proof}

In~\cite{erdos}, the authors construct an $\lfloor n/\varrho \rfloor$-by-$\varrho$ rectangular grid of vertices, and lines of cardinality $\varrho$, intersecting each column, are hyperedges. The set of `slopes' is restricted to a set $A$ satisfying Lemma~\ref{lemma-3PA}, so that the hypergraph has Berge girth larger or equal to 4, see~\cite{erdos} for more details. However, the hypergraph might have Berge girth larger or equal to 3 if we do not restrict the set of slopes, thus giving rise to good PIR codes.

By mapping the hypergraph $\mathcal{G}(V,E)$ constructed in~\cite{erdos} back onto its bipartite incidence graph $G(E,V,I)$,
and by using the notation for batch codes, we obtain a bipartite graph of girth at least $8$ with $(n-k)^{2-\epsilon}$ left vertices and $n-k$ right vertices. The corresponding graph-based asynchronous batch code has $k = (n-k)^{2-\epsilon}$, and so its redundancy is bounded from above by $\rho(k) = n - k =  
O\left({k}^{1/(2-\epsilon)}\right)$ for any $\epsilon>0$, and for any fixed $t \ge 3$. 

We note that the upper bound in~\cite{erdos} similarly yields the lower bound  
\begin{equation}
\lim_{k \rightarrow \infty} \frac{\rho(k)}{\sqrt{k}} \rightarrow \infty \; 
\label{eq:optimal_redundancy}
\end{equation}
for the optimal redundancy $\rho(k)$ of the graph-based asynchronous codes, for any fixed $t \ge 3$.

We compare these results with their counterparts for (non-asynchronous) batch codes in~\cite{Vardy}, where it was shown that 
for any $t \ge 5$ the optimal redundancy of general (multiset primitive) linear batch codes behaves as $O(\sqrt{k} \log k)$,
while for $t \in \{3,4\}$ the corresponding redundancy is $O(\sqrt{k})$. 
It is worth mentioning that for $t \in \{3,4\}$ there is a gap between the optimal redundancy $O(\sqrt{k})$
of the codes studied in~\cite{Vardy} and the lower bound~(\ref{eq:optimal_redundancy}) for the graph-based batch codes presented in this work. 
It remains an open question what is the exact asymptotics for the graph-based asynchronous batch codes for various values of $t$, 
and whether the lower bound~(\ref{eq:optimal_redundancy}) 
actually matches for some values of $t$ the upper bound $O(\sqrt{k} \log k)$ obtained in~\cite{Vardy}, 
or there is a gap between the optimal redundancy of these two families of codes.


\section{Properties and Bounds of $t=2$ Asynchronous Batch Codes}
\label{sec:properties}

In this section, we focus on a special case $t=2$. 
We first present simple lower and upper bounds on the redundancy for (general) asynchronous batch codes. 
In the second part of this section, we derive lower and upper bounds on the optimal redundancy of graph-based asynchronous batch codes and show their tightness. 
We also define an additional class of asynchronous batch codes, and derive bounds on their optimal redundancy. 

We start with the following two simple lemmas, in which we give estimates on the size of $\cA(k,t=2)$. 

\begin{lemma}
We have $\cA(k,t=2) \ge k+1$ for any $k \ge 1$. 
\label{lemma:k_1}
\end{lemma}

\begin{proof}
Assume that a pair of requests $(x_i, x_j)$ is to be served, where $i, j \in [k]$, and assume that $x_i$ has been served first. Denote by $V$ the vector space spanned by the columns of $\bldG$, which are not used for recovery of $x_j$. Since any of the 
symbols $x_1, x_2, \cdots, x_k$ should be recoverable from the remaining columns of $\bldG$, 
$V$ should contain the unity vectors $\blde_1, \blde_2, \cdots, \blde_k$, and therefore
its dimension is at least $k$. Therefore, there should be at least $k$ such columns. 
\end{proof}

\begin{lemma}
We have $\cA(k,t=2) \le \cA(k,t=2, r=2) \le 2k-1$ for $k \ge 3$. 
\label{lemma:2k-1}
\end{lemma}

\begin{proof}
Let $\bldy = \bldx \bldG$, where the vectors $\bldx$ and $\bldy$ have length $k$ and $n$, respectively. 
We construct $\bldG$ by defining the entries $y_i$, $i \in [n]$, of $\bldy$, for $n = 2k-1$. Specifically,  
\begin{eqnarray*}
\left\{ \begin{array}{lcl}
y_i & = & x_1 + \sum_{\ell = 3}^{i+1} x_\ell \quad \mbox{ for $i = 1, 2, \cdots, k-1$} \\
y_k & = & \sum_{\ell = 1}^{k} x_\ell \\
y_{2k - i} & = & x_2 + \sum_{\ell = 3}^{i+1} x_\ell \quad \mbox{ for $i = 1, 2, \cdots, k-1$}  
\end{array} \right. \; . 
\end{eqnarray*}
where the empty sum is assumed to be equal zero. 
Then, we have the following reconstruction sets of sizes $1$ and $2$, two recovery sets for each $x_i$: 
\begin{eqnarray*}
\left\{ \begin{array}{lcl}
x_1 & = & y_1 \; = \; y_k - y_{k+1} \\
x_2 & = & y_k - y_{k-1} \; = \; y_{2k-1} \\
x_i & = & y_{i-1} - y_{i-2}  = \; y_{2k - i + 1} - y_{2k - i + 2} \quad \mbox{ for $i = 3, \cdots, k$ } 
\end{array} \right. \; 
\end{eqnarray*}
Specifically, we see that for each $i \in [k]$, one copy of $x_i$ is recoverable from a subset of size at most two of $\{y_1, y_2, \cdots, y_k\}$, and a second copy 
of $x_i$ is recoverable from a subset of size at most two of $\{y_k, y_{k+1}, \cdots, y_{2k-1}\}$. Moreover, $y_k$ is used in recovery of one copy of $x_1$ and one copy of $x_2$ only. By checking all the possibilities, it is straightforward to verify that for each choice of a recovery set $S_i$ for $x_i$, $i \in [k]$, there exists a recovery set $S_j$  of $x_j$, for any $j \in [k]$, such that $S_i \cap S_j = \varnothing$. 
\end{proof}

Next, we focus on the graph-based asynchronous batch codes. 
For the sake of completeness of the discussion, we remark that the lower bound $\rho \ge \sqrt{2k} + O(1)$ on the optimal redundancy of PIR codes (for $t\ge 3$) 
was obtained by Rao and Vardy in~\cite{Rao-Vardy}. This result implies analogous lower bound on the redundancy 
of (regular) batch codes for $t\ge 3$. It is also shown that the bound is tight for PIR codes. 
Moreover, for $t=2$, the optimal redundancy for batch codes is just one bit~\cite{Vardy}.
By contrast, we show a lower bound $\rho \ge 2\sqrt{k}$ for graph-based asynchronous batch codes (for all $t \ge 2$), and present an explicit construction of asynchronous batch codes for $t = 2$ that attain this bound. 

\begin{theorem}
\label{thrm:mantel}
Let $\code$ be a graph-based asynchronous $[n,k,t \ge 2]$ batch code.
Then, its redundancy is $\rho \ge 2\sqrt{k}$. 
\end{theorem}

\begin{proof}
Let $\hat{G} = (A, B, \hat{E})$ be a bipartite graph that corresponds to the code $\code$. 
Then, the girth of $\hat{G}$ is $\ge 8$, and $\d(a) \ge 2$ for $a \in A$. 
Also, $k=|A|$ and $n-k=|B|$.

First, we delete edges of $\hat{G}$ such that after deletion $\d(a)=2$ for $a \in A$, and denote the new graph $G$ (note that we change the code). We construct a new (non-bipartite) graph, $G' = (V', E')$, from $G$, by following the correspondence in 
Definition~\ref{def:correspondence}. Since the left degree of $G$ is 2, the result is indeed a graph (rather than hypergraph).
Specifically, take $V' = B$. For each $u \in A$, 
replace $u$ and two edges $\{u, v_1\}$ and $\{u, v_2\}$ incident with it by a new edge $e_u = \{ v_1, v_2 \}$. 
The construction implies that there is a cycle of length $2t$ in $G$ if and only if there is a cycle of length $t$ in $G'$. 
Thus, $G$ has girth $\ge 8$ if and only if $G'$ has girth $\ge 4$. 

By Mantel's Theorem~\cite{mantel} (see also: Tur\'an's Theorem~\cite{turan}), this implies that the number of edges $|E'|$  satisfies $ |E'| \le |V'|^2/4 \; .$
Since $|A| = k$ and $|B| = n-k$, we obtain that $|V'| = n-k$ and $|E'| = k$. Therefore, the redundancy $\rho = n-k \ge 2\sqrt{k}$.
The redundancy of the original code is at least as large.
\end{proof}

This bound is in fact tight. 
\begin{example}
\label{example:tight-1}
Consider a complete bipartite graph $G'$ with a vertex set $V' = A' \cup B'$, $A' \cap B' = \varnothing$, $|A'| = |B'|$. This graph has $|V'|^2/4$ edges in total, and girth 4. Moreover, this graph has the largest possible  
number of edges for any girth-4 graph with $|V'|$ vertices, as seen by Mantel's Theorem~\cite{mantel}.

Next, we convert this graph into a bipartite graph $G$ 
by using the inverse of the above mapping. Namely, each edge is replaced 
by a triple ``edge, vertex, edge''.  We obtain that $G$ is a left regular bipartite graph of left degree 2 with $|A| = |V'|^2/4$ and $|B| = |V'|$. The graph $G$ has girth 8 and hence it yields an asynchronous batch code having length $n = |V'|^2/4+ |V'|$, number of information symbols $k=|V'|^2/4$, redundancy $\rho = 2\sqrt{k}=|V'|$, and $t=2$.
\end{example}


\bigskip

Next, we turn to defining another class of asynchronous batch codes. In the sequel, we derive bounds on their optimal redundancy. 

\begin{definition}
Assume that an asynchronous $[n,k,t]$-batch code $\code$ satisfies the following conditions: 
\begin{enumerate}
\item[C1.] The code is systematic.
\item[C2.] Every recovery set consists of either
  \begin{description}
	\item[Type (a):] $\quad$ one symbol from the systematic part, or
      \item[Type (b):] $\quad$ one symbol from the check part and some symbols from the systematic part.
  \end{description}
\end{enumerate}
Then, we call $\code$ a \emph{code satisfying conditions C1 and C2.}
\end{definition}

We remark that the codes satisfying conditions C1 and C2 are a variation of graph-based batch codes.
However, graph-based codes use only recovery sets of Type (b), while the codes satisfying conditions C1 and C2 can also use singleton recovery sets (the sets of Type (a)). 
 
As we show in the sequel, it is possible to show a tighter upper bound on the code redundancy when compared with its counterpart in 
Theorem~\ref{thrm:mantel}. Thus, in this case we show that the optimal redundancy is $\sqrt{2k} + O(1)$.  

Under conditions C1 and C2, we are able to find the shortest possible length of an asynchronous batch code with $t=2$. For $k \in \nn$, let $m_k$ be the smallest positive integer such that ${m_k \choose 2} \geq k$. This is equivalent to $m_k(m_k - 1) \geq 2k$, thus implying that $m_k \geq \sqrt{2k} + O(1)$. 

\begin{lemma}
\label{lemma:m_k}
$k-1+m_{k-1} \geq (k+m_k) - 2$
\end{lemma}

\begin{proof}
We need to prove that $m_{k-1} \geq m_k - 1$. The conclusion follows immediately by using an expression 
for the number of combinations ${m_k \choose 2}$.
\end{proof}

\begin{proposition}
\label{prop:shortest}
Let $\code$ be an asynchronous $[n,k,t=2]$ code satisfying conditions C1 and C2. Then the shortest possible length of $\code$ is $n = k+m_k$, i.e.~the optimal redundancy is $\rho = m_k$.
\end{proposition}

\begin{proof}
In light of Lemma~\ref{lemma:m_k}, we may assume that in the redundancy part of the generator matrix $\bldG$ there is no column of weight one. 

Indeed, assume to the contrary that there is a column $(1,0,\ldots,0)^T$ in the redundancy part of $\bldG$. Together with the systematic part, there are now two such columns. To recover the symbol $x_1$, we are now able to choose either or both of those columns: for this we need to show that we may assume none of these columns will have to be used to recover any other symbol. To achieve this, change the code by replacing any other ones in the first row by a zero, such that the weight of the first row is now two. At the same time note that any recovery set for any symbol other than $x_1$ can be modified by removing any of the two copies of the symbols corresponding to the column $(1,0,\ldots,0)^T$ from the recovery set. In effect this means that we split the recovery task into two disjoint recovery tasks: use any copies of $(1,0,\ldots,0)^T$ to recover symbol $x_1$ and use columns among the other columns to recover any other symbol. Thus, we can now use the relation $\cD(k,t=2) \leq \cD(k-1,t=2) + 2$,
where $\cD(k,t)$ denotes the minimal length of the code satisfying conditions C1 and C2 of dimension $k$ that supports $t$ queries. 
Therefore, we have reduced the problem to using a code for information vectors of length $k-1$.

It follows from Lemma~\ref{lemma:m_k}, by strong induction, that we may assume that in the redundancy part of $\bldG$ there is no column of weight-1. Now we formulate the main result of this section.

The algorithm for choosing the recovery set for the incoming request $x_i$ acts as follows.
\begin{itemize}
	\item If the column $i$ from the systematic part is available then choose this column as a singleton recovery set;
	\item otherwise choose an available recovery set of Type (b).
\end{itemize}

Under these assumptions we prove the following claims.

\begin{claim}
We may assume that the weight of each row in the generator matrix is at least three.
\label{claim:weight-three}
\end{claim}

Assume to the contrary that w.l.o.g.~the weight of the first row is two. It has to be at least two since each recovery set of $x_1$ uses at least one of the respective columns. Now, observe that the respective column in the redundancy part can only be used to recover the symbol $x_1$. Indeed, assume it recovers a symbol $x_i$, $i \neq 1$. This means that both columns whose first entry is $1$ are in use for this recovery task. When this recovery task is ongoing, no additional available recovery set can be found for the symbol $x_1$. This is a contradiction to the assumption that the code can satisfy any $t$ asynchronous requests, for $t = 2$. 

Next, note that we can modify the code by replacing the entries in the rows other than the first row in these two columns by zeros, and by changing the recovery sets for the symbol $x_1$ to be the singletons.
To this end, we can use Lemma~\ref{lemma:m_k} and the strong induction to conclude that we can assume that the weight of every row is at least three.



\begin{claim}
Suppose that the column $\bldb$ in the redundancy part of $\bldG$ is used to recover any of the distinct symbols $x_{i_1},x_{i_2},\ldots,x_{i_h}$ and no other symbols. Then we may assume that this column has weight $h$ and has ones precisely in positions $i_1,i_2,\ldots,i_h$.
\label{claim:use-symbol}
\end{claim}

To see this, note that the $i$-th column of $\bldG$, $\bldb$, is involved only in recovery sets of Type (b), where all the remaining columns of each recovery set come from the systematic part. It is clear that if the entry $g_{i_j,i}$ in row $i_j$ and column $i$ of $\bldG$ is equal to zero, then the symbol $x_{i_j}$ cannot be recovered using column $\bldb$, therefore $\bldb$ has ones in all the positions $i_1,i_2,\ldots,i_h$. The columns in the systematic part, which are used in each recovery set, that involves the column $\bldb$, are as follows: to recover $x_{i_j}$, the columns $i_1,\ldots,i_{j-1},i_{j+1},\ldots,i_h$ have to be used. On the other hand, if $\bldb$ had any additional ones, the respective recovery sets would be supersets of these. Therefore we may assume that all other entries in column $\bldb$ are zeros, possibly reducing recovery sets (throwing out columns) in the process.

In what follows we make the assumption of the preceding claim about the columns in the redundancy part.

In the following observation, we consider which symbols these columns are used to recover.

\begin{observation}
\label{observ:two-weight-2}
Two equal weight-2 columns can be replaced by weight-1 columns by replacing a single one by a zero in each of the corresponding rows. Thus, we may assume that this
case does not occur by the preceding Lemma and strong induction.
\end{observation}

\begin{observation}
\label{observ:three-rows}
Two columns which both have ones in the same three (or more) rows cannot be in use at the same time.
\end{observation}

\begin{observation}
Two columns for which there are exactly 2 rows $i$ and $j$ in which both have an one, can only be in use at the same time to recover the symbols $x_i$ and $x_j$.
\end{observation}

\begin{observation}
Suppose two of the same requests $x_i$ and $x_i$ come in after each other. If the first request is recovered by a set of Type (a), then the second one is recovered by Type (b). If the first is recovered by a set of Type (b), then the second is recovered by Type (a). Thus, in principle, we are allowed to fix such a recovery set of Type (b) for the first situation in the description of the algorithm, however, it is not necessary in our considerations.
\end{observation}

\begin{lemma}
Consider the set of columns in the redundancy part of $\bldG$. For a column of weight $h$, there are at least $h$ other columns in the redundancy part whose support intersects the support of the given column.
\label{lemma:h}
\end{lemma}

\begin{proof}
Suppose that the column is able to recover the distinct information symbols $x_{i_1},x_{i_2},\ldots,x_{i_h}$. Therefore its nonzero entries are exactly in positions $i_1,i_2,\ldots,i_h$. Suppose that, w.l.o.g., the column is in use to recover a symbol other than $x_{i_1}$, and that a new request for $x_{i_1}$ comes in
(more generally, we may have any symbol $x_{i_1},x_{i_2},\ldots,x_{i_h}$ in the role of $x_{i_1}$). 
Due to Observation~\ref{observ:three-rows}, there are two possibilities:
\begin{enumerate}[1)]
	\item There is a column in the redundany part whose support intersect the support of the given column only in position~$i_1$. 
	\item There is no column as in 1), but there are $h-1$ other columns whose supports intersect the given column in exactly two positions, say, $i_1$ and $i_2$, $i_1$ and $i_3$, $\ldots$, $i_1$ and $i_h$, respectively.
\end{enumerate}
Let us show that in each case the claim of the lemma holds.

\begin{itemize}
\item{If $h=2$}, then the only way the claim of the lemma may not hold is if there is another column whose support contains $\{i_1,i_2\}$, and there is no column in the redundancy part whose support intersect the support of the given column in a single position. If the support for that column equals $\{i_1,i_2\}$, we reach a contradiction by Observation~\ref{observ:two-weight-2}. Alternatively, its support includes another position, assume it is $i_3$. Now suppose that that column can be used for recovering symbol $x_{i_3}$ (otherwise, we have a contradiction to Claim~\ref{claim:use-symbol}).
In this case, the request for neither $x_{i_1}$ nor $x_{i_2}$ can be recovered: indeed, each of the remaining columns has the same symbol in positions $i_1$ and $i_2$ which means that only $x_{i_1} + x_{i_2}$ can be found but none of the individual bits. This is a contradiction. Therefore for $h=2$ at least two other columns have to exist whose support intersects the support of the given column.

\item{If $h\geq 3$}, we consider the following cases. If all positions $i_1,i_2,\ldots,i_h$ obey possibility 1), then the claim is true. Otherwise, w.l.o.g., let~$i_1$ be a position obeying possibility 2). This implies that there are at least $h-1$ columns whose support intersects the support of the given column. 

Let us show that there exist an additional column in the redundancy part of $\bldG$ whose support intersects the support of the given column. Indeed, if any of the positions $i_2,\ldots,i_h$ obeys possibility 1), we are done. Therefore, we may assume that they all obey possibility 2). This means, however, that the total number of columns whose support intersects the support of the given column is at least ${h\choose 2} \geq h$ for $h \ge 3$.
\end{itemize}
\end{proof}

We turn now to completion of the proof of Proposition~\ref{prop:shortest}. We do that by bounding from below the total number of unordered pairs of  different columns in the redundancy part of $\bldG$. Obviously, this number is larger or equal to the total number of pairs of columns whose supports intersect. By applying Lemma~\ref{lemma:h}, this number is larger or equal to the total weight of all columns in the redundancy part, divided by two. From Claim~\ref{claim:weight-three}, we have that this number is larger or equal to $2k/2 = k$. This completes the proof of the proposition.
\end{proof}

In the following example, we note that the bound in Proposition~\ref{prop:shortest} is tight for infinitely many values of $k$. 
\begin{example}
\label{example:tight-2}
Consider the asynchronous $[n,k,t=2]$-batch code $\code$ satisfying C1 and C2 with $n = k + m_k$, $k \ge 2$, whose generator matrix is defined as $\bldG = [ \; \bldI \; | \; \bldA \; ]$, 
and $\bldA$ is a $k \times m_k$ binary matrix that consists of all possible different rows of weight two, and $\bldy = \bldx \bldG$. 
It is straightforward 
to see that the following four cases hold: 
\begin{itemize}
\item[(a)]
$x_{i}$ is being served using a systematic part, and a new requests $x_{i}$ comes in; 
\item[(b)]
$x_i$ is being served using a redundancy part, and a new request $x_{i}$ comes in;
\item[(c)]
$x_i$ is being served using a systematic part, and a new request $x_{j}$, $i \neq j$, comes in;
\item[(d)]
$x_i$ is being served using a redundancy part, and a new request $x_{j}$, $i \neq j$, comes in. 
\end{itemize}

We note that the supports of any two columns intersect only in one position. Then, it is straightforward to check that in the cases (b) and (c), the new request can  always be served using a singleton recovery set in the systematic part. The request in the case (a) can be served using a recovery set in the redundancy part. 

In the case (d), if the information symbol $y_j = x_j$ is available, then it can be used for the recovery of the request. Otherwise, if 
the singleton $y_j$ is currently being used for recovering $x_i$, then the symbol corresponding to the other parity column, say $\bldg$, with one at position $j$, is free and can be used along with some singleton columns. Since supports of these two columns in the redundancy part intersect only in a single one in position $j$, the recovery set for $x_j$ is disjoint with the existing recovery set of $x_i$. 

We remark that it was shown in~\cite{Rao-Vardy} that the code $\code$ is also a non-asynchronous $[n = k + m_k, k, 3]$ batch code.  
\end{example}

Table~\ref{table:t_2} summarizes the lower and upper bounds on the redundancy for different models of codes. 

\begin{table*}[!h]
	\caption{Lower and upper bounds on the redundancy $\rho(k)$ for different models of asynchronous batch codes, $t=2$.}
	\label{table:t_2}
	\renewcommand{\arraystretch}{1.2}
	\begin{tabular}{|l|l|l|}
\hline
Code type & Lower bound & Upper bound \\
\hline
General asynchronous batch codes & $1$, Lemma~\ref{lemma:k_1} & $k-1$, Lemma~\ref{lemma:2k-1} \\
& & $\sqrt{2k} + O(1)$, Example~\ref{example:tight-2} \\
\hline
Asynchronous batch codes satisfying C1 and C2 & $\sqrt{2k} + O(1)$, Proposition~\ref{prop:shortest} & $\sqrt{2k} + O(1)$, Example~\ref{example:tight-2} \\
\hline
Graph-based asynchronous batch codes & $2\sqrt{k}$, Theorem~\ref{thrm:mantel} & $2\sqrt{k}$, Example~\ref{example:tight-1} \\
\hline
	\end{tabular}
	\renewcommand{\arraystretch}{1}
\end{table*}
We remark that the upper bound of $\sqrt{2k} + O(1)$ in Example~\ref{example:tight-2} provides a tighter upper bound on the optimal redundancy of general asynchronous batch codes than the counterpart in Lemma~\ref{lemma:2k-1} .

\section{Some Observations on $t>2$ Asynchronous Batch Codes}
\label{sec:general}
In this section, we present several additional results. 
We start with the following lemma. 

\begin{lemma}
Suppose that the generator matrix $\bldG$ of an $[n,k,t]$ batch code contains an $p \times s$ submatrix $\bldA$ consisting either of all zeros or of all ones, for some $s\in \nn$ and for $p \ge 2$. Then, $n \ge s + t$. 
\end{lemma}

\begin{proof}
Let $i$ and $j$, $i \neq j$, be indices of the rows of $\bldG$ which overlap with $\bldA$. Denote by $\cS_A$ and $\cT_A$ the subsets of rows and columns of $\bldG$, respectively, which overlap with $\bldA$. Consider two information vectors $\bldx, \hat{\bldx} \in \ff^k$, such that $x_i = x_j = 0$, $\hat{x}_i = \hat{x}_j = 1$, and $x_\ell = \hat{x}_\ell$ for $\ell \notin \cS_A \setminus \{i,j\}$. 
Denote $\bldy = \bldx \bldG$ and $\hat{\bldy} = \bldG \hat{\bldx}$.

Assume that $t$ users query $t$ copies of $x_i$. Then, for each copy of $x_i$, at least one symbol in the recovery set should differ in 
$\bldy$ and in $\hat{\bldy}$. Therefore,   
$\bldy$ and in $\hat{\bldy}$ must have at least $t$ different symbols outside the positions corresponding to $\cT_A$. We conclude that $n \ge s + t$. 
\end{proof}

\medskip

\begin{proposition}
Let $\bldG= [\bldI_k|\bldA_{k,{k \choose 2}}]$, where $\bldA_{k, {k \choose 2}}$ is a binary $k \times {k \choose 2}$ matrix with all possible columns of weight two, $k \ge 2$. Then $\bldG$ generates an asynchronous batch code which supports any $t=k-1$ queries. 
\end{proposition}

\begin{proof}
Let $\bldy =  (y_1, y_2, \cdots, y_n) = \bldx \bldG$. 
It is sufficient to restrict ourselves to the recovery sets that consist of two symbols, one symbol corresponds to a column of weight one in the systematic part of $\bldG$, and one symbol corresponds to a column of weight two in the redundancy part of $\bldG$. 
   
Assume that the $k-2$ requested symbols $(x_{i_1}, x_{i_2}, \ldots, x_{i_{k-2}})$ are being served, and an additional symbol $x_j$ is requested by the user. 
There are $k-1$ columns $\bldh$ of weight two in $\bldG$, which have $h_j = 1$. Since $k-2$ symbols are being served, then only $k-2$ symbols in the systematic part 
are used. There are two cases. 
\begin{description}
\item[Case 1]: if the symbol $y_j$ in the systematic part is not used, then there are at most $k-2$ other symbols in the systematic part, which are used. Let
$y_i$ be an unused symbol in the systematic part, $i \neq j$. Then, the symbol $y_\ell$ corresponding to the column $\bldh$ of weight two with $h_i = h_j = 1$ is also not used. We then use $x_j = y_i + y_\ell$. 
\item[Case 2]: if the symbol $y_j$ in the systematic part is used, then it is used together with a symbol $y_{j'}$ corresponding to a column $\bldh$ of weight two in the redundancy part, $h_j = 1$. In that case, additional $k-3$ symbols in the systematic part other than the symbol $y_j$ are used. Let $y_\ell$ be an index of an unused symbol in the systematic part. Consider the column $\bldg$ that has two nonzero entries $g_j$ and $g_\ell$. The symbol $y_{\ell'}$ corresponding to that column is not used since the symbol $y_\ell$ is not used, and the symbols $y_j$ and $y_{j'}$ are combined together. Then, we can use for the recovery $x_j = y_\ell + y_{\ell'}$. 
\end{description}
\end{proof}

The next result shows a way for constructing an asynchronous $[n',k,t=3]$ batch code, $n'=n+k$, from a given $[n,k,t=3]$ batch code by appending an identity matrix on the right of $\bldG$. 
\begin{proposition}
Let $\bldG= [ \; \bldI_k \; | \; \bldA \; ]$ be a $k \times n$ generator matrix of a systematic (non-asynchronous) $[n,k,t=3]$-batch code $\code$ that satisfies Conditions C1 and C2. 
Then $\bldG'= [ \; \bldI_k \; | \; \bldA \; | \; \bldI_k \; ]$ generates an asynchronous $[n',k,t=3]$-batch code, $n'=n+k$.
\end{proposition}
\begin{proof}
Let $\code$ be a batch code as in the condition of the proposition, and let $\bldG= [ \; \bldI_k \; | \; \bldA \; ]$ be its generator matrix. 
Let $\code'$ be a code generated by the matrix 
$\bldG'= [ \; \bldI_k \; | \; \bldA \; | \; \bldI_k \; ]$. 

Denote $\bldy = \bldx \bldG$. Assume that the value of symbol $x_i$ is requested. There are three non-overlapping recovery sets for $x_i$  in $\code$, denote them $R_{i,1}$, $R_{i,2}$ and $R_{i,3}$, where $R_{i,1}$ is a singleton in the systematic part. Additionally, there is a 
singleton recovery set in the block of the $k$ right-most positions of $\bldy$ (we call such recovery sets to be in the \emph{right part}). We use only one of these four subsets for the recovery of $x_i$. We assume the following algorithm for the recovery of the requested information symbols. 
\begin{enumerate}
\item
If the symbol $x_i$ is available in the systematic part of $\bldy$, use it as a singleton recovery set. 
\item
Otherwise, if the symbol $x_i$ is not available in the systematic part, use any of the recovery sets $R_{i,j}$, $j = 2,3$. 
\item
Otherwise, if the sets $R_{i,j}$, $j \in [3]$, are not available, use the singleton recovery set in the right part of $\bldy$.
\end{enumerate} 

We show that this algorithm always succeeds to satisfy any three requests in an asynchronous manner. We consider the following cases. 
\begin{itemize}
\item
If the three requested symbols are three copies of the same symbol, say $(x_i, x_i, x_i)$, then the four corresponding recovery sets are all disjoint, and therefore 
there are available recovery sets. 
\item
If the three requested symbols are $(x_i, x_j, x_\ell)$, where $i$, $j$ and $\ell$ are all different, then
assume that the requests arrive in the order $x_i$, $x_j$, $x_\ell$. Independently of the choice of the recovery set 
for $x_i$, there is always an available disjoint recovery set for $x_j$ in the right part of $\bldy$. 
Independently of the choice of the recovery sets for $x_i$ and $x_j$, there is always an available disjoint recovery set for $x_\ell$ in 
the right part of $\bldy$. 
\item
If the three requested symbols are $(x_i, x_i, x_j)$, $i \neq j$, then consider different orders of arrivals of these requests. 
\begin{itemize}
\item
If the first two requests are $x_i$ and $x_i$, then the third request $x_j$ can be recovered using the right part of $\bldy$. 
\item
If the first two requests are $x_i$ and $x_j$ in any order, and the third request is $x_i$, then assume to the contrary that 
the last $x_i$ cannot be recovered. In particular, this means that the symbol $x_i$ in the right part of $\bldy$ is used 
as a singleton recovery set, and that $x_j$ is recovered using the recovery set $S$ of Type (b), which contains $y_i$ in the 
systematic part of $\bldy$. 

The set $S$ should intersect both $R_{i,2}$ and $R_{i,3}$, otherwise we could use one of them for recovery of $x_i$.
Since $R_{i,2}$ and $R_{i,3}$ are disjoint, this means that $S$ is different from each of them. Therefore, $S$ is never used 
for recovery of $x_i$. However, since $S$ contains the systematic symbol $y_i$, but is used for recovery of $x_j$, then it contains some symbol $y_\ell$,
which corresponds to a column $\bldg$ of $\bldG$, which has $g_i = 1$. 
However, by Claim~\ref{claim:use-symbol}, it should hold $g_i = 0$, since column $\bldg$ is not used for the recover of $x_i$. We obtain a contradiction, 
thus completing the analysis of all possible cases.
\end{itemize}
\end{itemize}
\end{proof}

\begin{proposition}
\label{pro:asynch}
An $[n,k,t,r]$ batch code (with the restricted size $r$ of the recovery sets), $2 \le r < t$, is an $\left[n,k,\left\lfloor\frac{t}{r}\right\rfloor\right]$ asynchronous batch code.
\end{proposition}
\begin{proof}
Let $\code$ be an $[n,k,t,r]$ batch code, and $2\le r< t$. Then, for each information symbol $x_i$, $i=1,2, \cdots, k$, there exist 
at least $t$ disjoint recovery sets.

Consider a batch of requests of size $\left\lfloor t/r \right\rfloor$, denote it $\left(x_{i_1}, x_{i_2}, \cdots, x_{i_{\lfloor t/r \rfloor} }\right)$. 
Since there exist $t$ disjoint recovery sets for any choice of $t$ queries, we have a disjoint collection of recovery sets for this batch, say $R_{i_1}, R_{i_2}, \cdots, R_{i_{\left\lfloor t/r \right\rfloor}}$, where $|R_{i_j}| \le r$ for any $j \in [{\lfloor t/r \rfloor}]$. 

Next, assume that one of the symbols $x_{i_h}$, $h \in \left[ \left\lfloor t/r \right\rfloor \right]$, has been recovered, and that another query $x_\ell$, $\ell \in [k]$, comes in. We show that there exists a recovery set for $x_\ell$ which is pairwise disjoint with all the recovery sets $R_{i_j}$, $j \in \left[ \lfloor t/r \rfloor \right]$.

Let $R_{\ell_1}, R_{\ell_2}, \cdots, R_{\ell_t},$ be the $t$ disjoint recovery sets for $x_\ell$. Since $|R_{i_j}|\le r $, each $R_{i_j}$, $j \in \left[ \left\lfloor t/r \right\rfloor \right]$, has a nonempty intersection with at most $r$ recovery sets $R_{\ell_i}$, $i \in [t]$. Then, the maximum number of recovery sets for $x_\ell$, which overlap with the currently used recovery sets $R_{i_j}$, is $r \cdot (\left\lfloor t/r \right\rfloor-1)$. Therefore, there exists at least 
$$t-r \cdot \left(\left\lfloor\frac{t}{r}\right\rfloor-1\right) \ge t-r \cdot \left(\frac{t}{r}-1\right)=r>1$$
recovery sets for $x_\ell$, which do not intersect the currently used sets $R_{i_j}$. Hence, $\code$ is an asynchronous $\left[n,k,r, \left\lfloor t/r \right\rfloor \right]$-batch code.
\end{proof}

\section{Conclusions}

In this work, we considered a model of distributed data storage system employing batch codes for improved load balancing. 
We introduced a concept of asynchronous batch codes, which are suitable for serving the user requests immediately upon their arrival. 
We showed that hypergraphs of Berge girth at least 4 yield asynchronous batch codes, which we called \emph{graph-based}. 
We derived lower and upper bounds on the optimal redundancy $\rho(k)$ of various types of asynchronous batch codes of dimension $k$
with the query size $t=2$. For a general fixed value of $t \ge 3$, we showed that the optimal redundancy of graph-based 
asynchronous batch codes is $\rho(k) = O\left({k}^{1/(2-\epsilon)}\right)$ for any small $\epsilon>0$, and that 
$\lim_{k \rightarrow \infty} \rho(k)/\sqrt{k} = \infty$.  
 
This work poses a number of open questions. Below, we list some of them:
\begin{enumerate}
\item
What is the optimal value of redundancy $\rho(k)$ for asynchronous (graph-based or non-graph-based) batch codes of dimension $k$, for specific values of $t \ge 3$?
\item
What is the effect of the maximum size of the recovery set $r$ on the redundancy? 
\item
What are the shortest length asynchronous batch codes for specific small values of $k$? 
\item
Given a non-asynchronous batch code of a certain length, what is the minimum increase in redundancy required to obtain an asynchronous batch code with the same parameters?  
\item
Find optimal and sub-optimal constructions of asynchronous batch codes. 
\item
Propose batch codes and algorithms that allow for efficient recovery of the requested symbols in the practical settings.
\end{enumerate}








\end{document}